\documentclass[conference,twocolumn,final]{IEEEtran}
\makeatletter
\def\ps@headings{%
\def\@oddhead{\mbox{}\scriptsize\rightmark \hfil \thepage}%
\def\@evenhead{\scriptsize\thepage \hfil \leftmark\mbox{}}%
\def\@oddfoot{}%
\def\@evenfoot{}}
\makeatother 

\usepackage[final]{graphicx}
\usepackage[reqno]{amsmath}
\usepackage{amssymb}
\usepackage{cite}
\usepackage{balance}
\usepackage{color}
\usepackage{caption}

\usepackage[german,vlined,boxed]{algorithm2e}
\newenvironment{algorithmic}{%
\algorithm
}{%
\endalgorithm
}
\usepackage{floatrow}

\newfont{\bbb}{msbm10 scaled 500}

\newfont{\bb}{msbm10 scaled 1100}

\newcommand{\argmax}{\operatornamewithlimits{argmax}}

\newtheorem{theorem}{Theorem}
\newtheorem{lemma}[theorem]{Lemma}
\newtheorem{definition}[theorem]{Definition}

\IEEEoverridecommandlockouts

\title{Throughput Optimal Scheduling with Dynamic Channel Feedback}

\author{\IEEEauthorblockN{Mehmet Karaca\IEEEauthorrefmark{0},
Yunus Sarikaya\IEEEauthorrefmark{0}, Ozgur Ercetin\IEEEauthorrefmark{0}, Tansu Alpcan, Holger Boche\\
 }
}

\begin{document}
\maketitle


\begin{abstract}
It is well known that opportunistic scheduling algorithms are
throughput optimal under full knowledge of channel and network
conditions. However, these algorithms achieve a hypothetical
achievable rate region which does not take into account the overhead
associated with channel probing and feedback required to obtain the
full channel state information at every slot. We adopt a channel
probing model where $\beta$ fraction of time slot is consumed for
acquiring the channel state information (CSI) of a single channel.
In this work, we design a joint scheduling and channel probing
algorithm named SDF by considering the overhead of obtaining the
channel state information. We first analytically prove SDF algorithm
can support $1+\epsilon$ fraction of of the full rate region
achieved when all users are probed where $\epsilon$ depends on the
expected number of users which are not probed. Then, for homogenous
channel, we show that  when the number of users in the network is
greater than 3,  $\epsilon > 0$, i.e., we guarantee to expand the
rate region. In addition, for heterogenous channels, we prove the
conditions under which SDF guarantees to increase the rate region.
We also demonstrate numerically in a realistic simulation setting
that this rate region can be achieved by probing only less than 50\%
of all channels in a CDMA based cellular network utilizing high data
rate protocol under normal channel conditions.
\end{abstract}


\section{Introduction}
\label{sec:intro} Scheduling is an essential problem for any shared
resource. The problem becomes more challenging in a dynamic setting
such as wireless networks where the channel capacity is time varying
due to multiple superimposed random effects such as mobility and
multipath fading. Opportunistic scheduling has emerged as an
attractive solution for improving the efficiency of wireless systems
with limited resources such as frequency band and power. In
principle, \textit{opportunistic} policies schedule the user with
the favorable channel conditions to increase the overall performance
of the system \cite{Liu:A03}. Optimal scheduling in wireless
networks has been extensively studied in the literature under
various assumptions. The seminal work by Tassiulas and Ephremides
have shown that a simple Lyapunov-based opportunistic algorithm that
schedules the user with the highest queue backlog and transmission
rate product at every time slot, can stabilize the network, whenever
this is possible~\cite{MW}.

There has been much work in developing scheduling algorithms for
down-link wireless systems for various performance metrics that
include stability, utility maximization and energy
minimization~\cite{MW},~\cite{Liu:A03}, ~\cite{Neely:energy06}.
However, the common assumption in these studies is that the {\em
exact} and {\em complete} channel state information, (CSI) of all
users is available at every time slot. Hence, these algorithms
achieve a \textit{hypothetical} rate region by assuming that full
channel state information is available without any channel probing
or feedback costs. However, in practice acquiring CSI introduces
significant overhead to the network, since CSI is obtained either by
probing the channel or via feedback from the users. In current
wireless communication standards such as WiMax~\cite{wimax} and
LTE~\cite{LTE}, there is a feedback channel used to relay CSI from
the users to base station. Obviously, this feedback channel is
bandlimited and it is impossible to obtain CSI from all users at the
same slot. As a motivating example, consider CDMA/HDR (High Data
Rate) system~\cite{HDR1}, where the Signal-to-Noise Ratio (SNR) of
each link is measured. The value of the SNR is then mapped to a
value representing the maximum data rate that can support a given
level of error performance. This channel state information is then
sent back to the base station via the reverse link data rate request
channel (DRC). The channel state information is updated every 1.67
ms and sent back as 4 bits. Assume that there are 25 users in a
cell, and hence, 100 bits of channel information has to be sent
every 1.67 ms. This requires 60Kbps of channel rate to be dedicated
only for channel measurements. The minimum data rate of HDR system
is 38.4 Kbps and the average data rate is 308Kbps. Thus, overhead of
acquiring CSI is twice the minimum data rate, and is approximately
more than 20$\%$ of the average transmission rate. This overhead
increases significantly in a multichannel communication system such
as LTE. Thus, channel probing must be done efficiently in order to
balance the trade-off between being opportunistic (e.g., obtaining
useful channel information) and consuming valuable resources.

Most of real systems run under limited hardware and software
capacity. Hence, in addition to taking into account the probing
cost, the implementability of feedback schemes must be ensured  in
practical systems. For instance, a practical feedback scheme must
not require high computation time and overhead for probing users. In
addition, any feedback algorithm must not use any statistics such as
channel or arrival statistics for probing decision  since in real
word, such statistics change over time. It is also beneficial to
design a limited feedback system that can operate over a wide
channel condition and numbers of users.

In this work, we consider in fully connected network (e.g., Cellular
network, WLAN) where  base station (BS) is transmitting to a fixed
number of users. We assume that each user has infinite buffer
capacity and data arrives into users' queues according a stochastic
arrival process. We aim to  develop a joint scheduling and channel
probing algorithm  that stabilizes users' queues with  taking into
account to the probing cost. Our algorithm dynamically determines
the set of channels that must be probed at every time slot based on
the information obtained from the channel state information of the
user that has the maximum queue length at a given time slot. The
scheduling part is based on well known Max-Weight
algorithm~\cite{MW}.

Our contributions are summarized as follows:
\begin{itemize}
\item We first propose a joint scheduling and channel probing
algorithm which is easy to implement and low cost. In addition, our
algorithm does not require any statistics (i.e., channel or arrival
statistics) and can be applied over correlated and even
non-stationary channels.
\item Technical contributions of our paper are as follows: we first show that the proposed algorithm can achieve
$1+\epsilon$ fraction of the full rate region of the case when all
users are probed. For homogenous channels, we analytically show that
$\epsilon > 0$ (e.g., we guarantee to increase the rate region) when
the number of users is greater than 3. For heterogenous channels, we
prove that $\epsilon > 0$ as the number of channel states goes
infinity or the number of users is large enough.
\item We implement a realistic network setting where we simulate High Data Rate (HDR) protocol in CDMA cellular
networks  and show by numerical analysis that when our proposed
algorithm is used a comparable performance with Max-Weight algorithm
with full CSI can be achieved by only probing on the average less
than 50\% of users.
\end{itemize}
The rest of the paper is organized as follows. Section
\ref{sec:related} summarizes the literature on opportunistic
scheduling algorithms considering probing overhead. Section
\ref{sec:model} presents the network model, the basic structure of
the channel probing model. In Section \ref{sec:sdf}, we give the
problem formulation under the proposed channel probing models.  In
Section \ref{sec:sdf}, we give joint scheduling channel probing
algorithms for both channel probing models. Numerical results are
presented in Section \ref{sec:sim}. Finally, Section
\ref{sec:conclusion} concludes the paper and presents possible
future directions.
\section{Related Works}
\label{sec:related} There has been a significant recent interest in
applying Max-Weight-like algorithms in more realistic wireless
network settings. Throughput-maximizing scheduling has been studied
with different forms of limited CSI. For instance, infrequent
channel state measurements was investigated in~\cite{Kar:infocom07}
and it was shown that achievable rate region shrinks as the frequent
of CSI decreases. The impact of delayed CSI was investigated
in~\cite{Ying:ITS08}. However, unlike these works, we aim to
minimize the number of probed users at a time while maximizing the
achievable rate region. In this context, one of the most intriguing
research challenges in the context of wireless networking with
limited feedback is the design of a scheduling policy that

\begin{enumerate}
\item is implementable, simple and low-complexity,

\item achieves high performance, i.e., low packet delay and large rate
region.

\item works (correlated or even over non-stationary channels) without requiring any statistics such as channel
distribution.

\end{enumerate}
In the following, we classify the works along this direction. \\
\textbf{Limited feedback bandwidth}: In~\cite{Gopalan:allerton07},
the authors proposed a joint scheduling and channel probing
algorithm stabilizing the network by allowing the base station to
probe  a subset of channels (or links) at each time slot. First, the
a throughput-optimality algorithm was developed that uses the
expected channel rates, i.e., the channel distributions are given.
While it is reasonable to estimate the joint channel state
distribution when channels are independent, when the number of
channels is small and channel statistics do not change. However,
such estimation becomes intractable in real word where channels are
non-stationary process and  when the number of channels is large.
Then, the authors proposed only queue length based algorithm.
However, throughput-optimality of that algorithm can only be shown
under certain condition, i.e., when channels are symmetric and
subsets of channels are disjoint. In addition, delay performance of
these algorithms are unknown. In ~\cite{Gopalan:infocom12}, the
authors proposed a variant of the algorithm
in~\cite{Gopalan:allerton07} to analyzes queue-overflow performance
for scheduling with limited CSI.

In~\cite{Ouyang:mobihoc11}, a feedback allocation algorithm was
proposed for multi-channel system with limited feedback bandwidth.
In other words, only a limited number of users can be probed at a
time. It was shown that the proposed algorithm can achieve
$1-\epsilon$ fraction of the full rate region when channel
distributions are known at the BS. Note that the algorithms
in~\cite{Gopalan:allerton07},~\cite{Ouyang:mobihoc11},~\cite{Gopalan:infocom12}
cannot achieve all three properties given above.

Unlike the works
in~\cite{Gopalan:allerton07},~\cite{Ouyang:mobihoc11},~\cite{Gopalan:infocom12},
we assume a more flexible channel probing model  where BS may probe
all users at a given time. In other words, we did not assume
bandwidth limited feedback channel. Similar model was used
in~\cite{Chang:mobicom07}. However, it was assumed that  all users
have saturated buffers, i.e., they always have packets in the
corresponding buffers and, hence, the network stability problem was
not investigated. In this sense, the most related work is
~\cite{Chaporkar:mobihoc09}, where the optimal feedback- scheduling
scheme for a single-channel downlink is derived. Specifically,
in~\cite{Chaporkar:mobihoc09}, it was assumed that a single channel
probing requires a certain portion of data transmission (i.e.,
$\beta$ fraction of data slot). Hence, the server has a cost for
probing channels and gains a reward (queue weighted throughput)
which depends on the user. The problem of finding optimal joint
algorithm is transformed into an optimal stopping time problem and
is solved by Markov Decision Process (MDP) where  the channel
probabilities are known to the BS. In addition, since the authors
uses MDP to solve the problem, this formulation is computationally
intractable as it
involves a high dimensional state. Therefore, this work does not satisfy all three properties either. \\
\textbf{Markovian Channels}: Studies
in~\cite{Quyang:Allerton10},~\cite{Li:Mobilecomp10},~\cite{Cem:infocom12},~\cite{Quyang:infocom11},~\cite{NeelyLearning}
attempt to learn the channel distribution for scheduling.
In~\cite{NeelyLearning}, a 2-stage decision procedure is used where
channel measurement and packet transmission has to be performed at
every time slot without knowing the channel distribution. This work
attempts to learn transmission success probabilities by averaging
the previous outcomes. Similarly, in~\cite{Quyang:Allerton10}, the
authors proposed to estimates the channel statistics by using some
portion of the time slots for observation slot with some
probability. In~\cite{Li:Mobilecomp10}, the authors proposed a
scheduling algorithm  under imperfect CSI in single-hop networks
with  i.i.d. channels. They consider that probing a channel brings a
certain amount of energy cost. Under this setup, the proposed
scheduling algorithm which decides whether to probe the channel with
the energy cost is a Max-Weight type scheduling policy that
minimizes the energy consumption subject to queue stability.
In~\cite{Quyang:infocom11}, it was assumed that wireless channels
evolve as Markov-modulated ON/OFF processes. With this assumption, a
exploitation-exploration trade-off was investigated.
Similarly,in~\cite{Cem:infocom12}, a two state discrete time Markov
chain with a bad state which yields no reward and a good state which
yields reward was considered. Aforementioned works assume that the
underlaying stochastic process of the channel evolves according to a
fixed stationary process such as ergodic Markov chain. In practice,
such an assumption does not hold most of the time. For instance, the
measurement study in~\cite{Bozidar:nonstationary11} shows that the
wireless channel exhibits time-correlated and non-stationary behavior. Thus, these works
cannot achieve all three properties given above.\\



\section{System Model}
\label{sec:model} We consider a cellular system with a single base
station transmitting to $N$ users with a fixed power. Let
$\mathcal{N}$ denote the set of users in the cell. Time is slotted,
$t\in \{0,1,2,\ldots\}$, and wireless channel between the base
station and a mobile user is assumed to be independent across users
and slots. The gain of the channel is constant over the duration of
a time slot but varies between slots. In practical systems (e.g.,
CDMA/HDR system \cite{HDR1}), transmission rate is determined by a
link adaptation algorithm, which selects the highest transmission
rate to meet a given allowable target error probability. Only finite
number of transmission rates can be supported due to modulation and
coding schemes. We assume that each channel has $L$ possible states
with corresponding rates $\mathcal{R}=\{r_1,r_2,\ldots,r_L\}$ listed
in descending order, i.e., $r_k> r_l$ if $l>k$. Note that $r_k$, $k
\in \{1,2,\ldots,\L\} $ only depends on the Signal-to-Noise Ratio
(SNR). We denote $R_n(t) \in \mathcal{R}$ as the channel state
information (CSI) of user $n$ at time $t$. For user $n$, let
$p_{k}^n$ denote the probability of being state $k$, i.e,
$Pr[R_n(t)=r_k]=p_{k}^n$ and each channel state is possible with
non-zero probability such that $p^{min} < p_{k}^n < p^{max}$,
$\forall n,i$.

\subsection{Channel Probing Model}
In this work, we assume that there is no dedicated feedback channel,
and CSI is relayed over the data channel. Hence, depending on the
scheduling algorithms, the scheduler can probe any number of users
at any time slot. We quantify the overhead of obtaining the CSI of a
single user in terms of a time fraction of the time slot as
in~\cite{Chaporkar:mobihoc09}. This time duration may include the
time spent for pilot signal transmission, measurement of the signal
strength of pilot signal and the transmission of CSI to the base
station.

We assume that at the beginning of a time slot, the BS sends a pilot
signal. Based on the quality of the received pilot signal, each user
$n$ can determine its current channel state $R_n(t)$. We also assume
that downlink and uplink channels are identical, i.e., the maximum
transmission rates in both directions are the same. Hence, the BS
obtains reports from a selected number of users before scheduling a
downlink transmission. Let $\mathcal{N}_p(t) \subset \mathcal{N}$ be
the set of users who report back their CSI. Let $I_n(t)$ be an
indicator function showing whether user $n$ is scheduled to be
transmitted to in slot $t$:
\begin{align}
I_n(t) =& \left\{ \begin{array}{l l}
                    1              & \text{; if user $n \in \mathcal{N}_p(t)$ and it is
                    scheduled
                                            }\\
                    0        & \text{; otherwise}
                \end{array} \label{eq:sch_des}
    \right.
\end{align}
Note that for successful transmission, a user scheduled to be
transmitted to in slot $t$ should be selected from among the users
who have reported their channel states.


We assume that $\beta$ fraction of the time slot is consumed to
obtain CSI from a single user. Hence, only $(1-\beta N_p(t))\times
T_s$ seconds are available for data transmission where $N_p(t)$  is
the cardinality of $\mathcal{N}_p(t)$ and  $T_s$ is the duration of
the time
slot. We assume that $1-\beta N_p(t) < 1$
Note that when full CSI is obtained, this time is equal to $(1-\beta
N)\times T_s$. We assume that at most one user can be scheduled at a
given time. Hence, the amount of data that can be transmitted to
user $n$ by the BS when $N_p(t)$ users are probed at time $t$ is
given by,
\begin{align}
D_{n}(t)=(1-\beta N_p(t)) T_s R_n(t)I_n(t). \label{eq:D}
\end{align}
We assume that $T_s$ is normalized to unit slot length, i.e.,
$T_s=1$ in the rest of the paper. Let $A_n(t)$ be the amount of data
(bits or packets) arriving into the queue of user $n$ at time slot
$t$. We assume that $A_n(t)$ is a stationary process and it is
independent across users and time slots. We denote the arrival rate
vector as $\boldsymbol
\lambda=(\lambda_1,\lambda_2,\ldots,\lambda_N)$, where $\lambda_n =
\textbf{E}[A_n(t)]$. Let $\boldsymbol
Q(t)=(Q_1(t),Q_2(t),\ldots,Q_N(t))$ denote the vector of queue
sizes, where $Q_n(t)$ is  the queue length of user $n$ at time slot
$t$.
\begin{definition}
A queue is strongly stable if
\begin{equation}
\limsup_{t\rightarrow \infty}\frac{1}{t}
\sum_{\tau=0}^{t-1}\textbf{E}[Q_n(t)] < \infty \label{eq:defination}
\end{equation}
\end{definition}
Moreover, if every queue in the network is stable then the network
is called stable. The dynamics of the queue of user $n$ is given by,
\begin{align}
Q_n(t+1)=[Q_n(t)+A_n(t)-D_{n}(t)]^+ .
\end{align}
where $[x]^+=\max(x,0)$.

Before giving the proposed scheduling algorithm the following
definitions are useful.

The \textit{achievable rate region} or rate region of a network is
defined as the closure of the set of all arrival rate vectors
$\boldsymbol \lambda$ for which there exists an appropriate
scheduling policy that stabilizes the network.
\begin{definition}
$\Lambda_{h}$ is the hypothetical rate region where full CSI is
available (e.g. by an Oracle) without any channel probing or
feedback costs, i.e., $\beta=0$.
\end{definition}
As discussed in \cite{rethinking}, it is impossible to achieve the
boundary of $\Lambda_{h}$ in real systems, since there is an
overhead of acquiring CSI from the users.
\begin{definition}
$\Lambda_{f}$ is the achievable rate region when probing cost is
taken into account and when all users' channels are probed at every
time slot according to the feedback model. Note that $\Lambda_{f}
\subset \Lambda_{h}$
\end{definition}
\begin{definition}
$\Lambda_{a}$ is the  algorithm-dependent achievable rate region
where the probing cost is taken into account.
\end{definition}

Note that if $\Lambda_{f} \subset \Lambda_{a} \subseteq
\Lambda_{h}$, then we can argue that there exists a solution or
algorithm which is more efficient than full CSI Max-Weight algorithm
in terms of the rate region. We define the \textit{weighted rate} of
user $n$ as follows:
\begin{align}
W_n(t)=Q_n(t)R_n(t),\label{eq:wrate}
\end{align}
and $w(t)$ is the maximum weighted rate such that $w(t)=\argmax_{n
\in \mathcal{N}} \ \{W_n(t)\}$. In the next section, we propose a
scheduling and dynamic feedback algorithm which finds the maximum
weighted rate at every time slot by probing only $N_p(t)$ number of
users where $N_p(t) \leq N$.
\section{Scheduling and dynamic Feedback (SDF) Algorithm }
\label{sec:sdf} Recall that our aim is to find a joint scheduling
and channel probing algorithm which is able to find the user which
has the maximum weighted throughput at each time slot with minimum
number of channel probing. In this section, we propose a joint
scheduling and dynamic feedback allocation algorithm that determines
the user which has the maximum weighted throughput at each time slot
without probing every user.\\

\begin{algorithmic}
\textbf{ SDF Algorithm}:\\
\textit{(1) probing decision}:\\
\begin{itemize}
\item  Step 1:  Determine the user which has the
maximum queue length,
\begin{align*}
i^*\triangleq\argmax_{i\in\mathcal{ N}}\{Q_i(t)\}
\end{align*}
\item Step 2:  Probe and acquire the CSI of user $i^*$.
Let  $R_{i^*}(t)$ be the CSI
of user $i^*$ at time $t$.\\
\item Step 3:  Broadcast the value of $R_{i^*}(t)$.

\item Step 4: The users which have higher rate than $R_{i^*}(t)$ report their
CSIs to BS.

Let us define the set $\mathcal{S}_p(t)$ as follows:
\begin{align*}
\mathcal{S}_p(t)\triangleq \{j:R_j(t) > R_{i^*}(t)\}
\end{align*}
\end{itemize}
\textit{(2) scheduling decision}:\\
BS schedules a user which has the maximum weighted throughput
according to Max-Weight algorithm~\cite{MW} as follows;
\begin{align}
n^*=\argmax_{n \in \mathcal{N}_p(t)}\{(1-\beta N_p(t))W_n(t)\}
\end{align}
where
\begin{align*}
\mathcal{N}_p(t) \triangleq \mathcal{S}_p(t)  \cup i^*
\end{align*}
\end{algorithmic}

\vspace{5 mm} \textbf{Intuition:} Let us assume that  user $n^*$ has
the maximum weighted rate at a time slot. Given CSI of the user
$i^*$, the CSI of users with rate lower than $R_{i^*}$ need not be
collected since their weighted rates are always smaller than that of
user $i^*$. SDF algorithm is especially efficient when the number of
users is large since in that case the number of users with rate
lower than that of user $i^*$ is large with high probability.

\subsection{Analysis of SDF Algorithm with Homogenous Channels}
Now, we investigate the increase in achievable rate region when SDF
algorithm is employed. We first consider a homogenous channel model
where,
\begin{align*}
p_{k}^n=p_k, \forall n.
\end{align*}
Let $M(t)$ denote the number of users which \textit{do not send}
their CSIs since their channel conditions are worse than the user
which has the maximum queue length at time $t$. The number of
fractions of time slot consumed for probing with SDF algorithm is
determined as follows: first, the BS acquires the CSI of user $i^*$
and $\beta$  fraction of time slot is used for probing user $i^*$.
Then, BS broadcasts the value of CSI of user $i^*$. We assume that
broadcasting  CSI of a user also consumes $\beta$ fraction of time
slot. Then, the number of users which have higher rate than
$R_{i^*}(t)$ is equal to $N-1-M(t)$. Hence, the total number of
fractions used for channel probing within SDF algorithm at time $t$
is given as follows:
\begin{align}
N_p(t)=1 + 1 + N-1-M(t)=N+1-M(t).\label{eq:Np}
\end{align}
Note that $N_p(t)=N$ when all users are probed as in conventional
Max Weight algorithm.

We consider the following two functions:
\begin{align*}
f_{s}(\boldsymbol Q(t))&={\textbf{E}}\left[\sum_{n \in \mathcal{N}_p(t)} (1-\beta N_p(t))W_n(t)I_n(t)|\boldsymbol Q(t)\right],\\
f_{m}(\boldsymbol Q(t))&={\textbf{E}}\left[\sum_{n \in \mathcal{N}}
(1-\beta N)W_n(t)I_n(t)|\boldsymbol Q(t)\right],
\end{align*}
where the expectation is taken with respect to the randomness of
channel variations and scheduling decisions. Given $\boldsymbol
Q(t)$, both Max-Weight algorithm with full CSI and SDF schedules the
same user which has the maximum weighted rate at every time slot.
Hence, the value of $W_n(t)I_n(t)$ is the same for both functions,
and the only difference between $f_{s}(\boldsymbol Q(t))$ and
$f_{m}(\boldsymbol Q(t))$ appears in the number of users probed.
Next,we analyze the performance of SDF algorithm in terms of
achievable rate region by using the theorem given
in~\cite{Eryilmaz:ToN05}.

\begin{theorem}~\cite{Eryilmaz:ToN05}
If for some $\epsilon > 0$ SDF algorithm guarantees
\begin{align*}
f_{s}(\boldsymbol Q(t)) \geq (1+\epsilon) f_{m}(\boldsymbol Q(t))
\end{align*}
for all $\boldsymbol Q(t)$, then SDF can achieve  $1+\epsilon$
fraction of the rate region $\Lambda_{f}$.
\end{theorem}

\begin{theorem}
\label{theorem:1} SDF algorithm can support $(1+\epsilon)$ fraction
of the rate region $\Lambda_f$, where
\begin{align}
\epsilon = \frac{\beta \left( \textbf{E}[M(t)]-1]\right)}{1-\beta
N}.\label{eq:epsilon}
\end{align}
\end{theorem}
\begin{proof}
By using \eqref{eq:Np}, $f_s(\boldsymbol Q(t))$ can be rewritten as
follows:
\begin{align}
&f_{s}(\boldsymbol Q(t))=\notag\\
&{\textbf{E}}\left[\sum_n
(1-\beta(N+1-M(t)))W_n(t)I_n(t)|\boldsymbol
Q(t)\right] \notag\\
&=f_{m}(\boldsymbol Q(t)) + {\textbf{E}}\left[\sum_n (\beta
M(t)-\beta)W_n(t)I_n(t)|\boldsymbol Q(t)\right]
\end{align}

Now, we consider the value of $f_s(\boldsymbol Q(t)) /
f_m(\boldsymbol Q(t))$ such that,
\begin{align}
&f_s(\boldsymbol Q(t)) / f_m(\boldsymbol Q(t))=\notag\\
&\frac{f_{m}(\boldsymbol Q(t)) + {\textbf{E}}\left[\sum_{n \in
\mathcal{N}_p(t)} (\beta M(t)-\beta)W_n(t)I_n(t)|\boldsymbol
Q(t)\right]}{f_m(\boldsymbol Q(t))}\notag\\
&=1+ \frac{{\textbf{E}}\left[\sum_n (\beta
M(t)-\beta)W_n(t)I_n(t)|\boldsymbol Q(t)\right]}{f_m(\boldsymbol
Q(t))}
\end{align}
where,
\begin{align}
&{\textbf{E}}\left[\sum_{n \in \mathcal{N}_p(t)} (\beta
M(t)-\beta)W_n(t)I_n(t)|\boldsymbol
Q(t)\right]=\notag\\
& \left[\sum_{m=0}^{N-1}(\beta m-\beta) {\textbf{E}}\left[\sum_{n
\in \mathcal{N}_p(t)} W_n(t)I_n(t)|\boldsymbol Q(t), M(t)=m\right]
\right]\notag\\
&\times \Pr[M(t)=m]\label{eq:c1}
\end{align}
Note that
\begin{align*}
{\textbf{E}}\left[\sum_{n \in \mathcal{N}_p(t)}
W_n(t)I_n(t)|\boldsymbol Q(t), M(t)=m\right]=\frac{f_m(\boldsymbol
Q(t))}{1-\beta N}
\end{align*}
Hence, \eqref{eq:c1} can be rewritten as follows:
\begin{align*}
\eqref{eq:c1}&=\frac{f_m(\boldsymbol Q(t))}{1-\beta N}
\sum_{m=0}^{N-1} \left[(\beta m-\beta)
\right] \Pr[M(t)=m]\notag \\
&=\frac{\beta f_m(\boldsymbol Q(t))\left(
\textbf{E}[M(t)]-1]\right)}{1-\beta N}
\end{align*}
Thus we have,
\begin{align}
f_s(\boldsymbol Q(t)) / f_m(\boldsymbol Q(t))\geq 1+\frac{\beta
\left( \textbf{E}[M(t)]-1]\right)}{1-\beta N}
\end{align}
Hence, the proposed algorithm can support $(1+\epsilon)$ fraction of
the rate region $\Lambda_f$ where $\epsilon = \frac{\beta \left(
\textbf{E}[M(t)]-1]\right)}{1-\beta N}$. This completes the proof.
\end{proof}

Let us define
\begin{align}
f_{h}(\boldsymbol Q(t))&={\textbf{E}}\left[\sum_{n \in
\mathcal{N}(t)} W_n(t)I_n(t)|\boldsymbol Q(t)\right]
\end{align}
Note that  $f_{h}(\boldsymbol Q(t))$ represents the hypothetical
capacity region, $ \Lambda_{h}$ since $\beta=0$. Hence, by comparing
$f_{h}(\boldsymbol Q(t))$ and $f_m(\boldsymbol Q(t))$, we can find
$\epsilon_{max}$ which the maximum fraction that can be supported.

\begin{lemma}
\label{lemma:0} The maximum fraction $\epsilon_{max}$ is given by,
\begin{align}
\epsilon_{max}=\frac{\beta N}{1-\beta N}
\end{align}
\end{lemma}

\begin{proof}
The proof is similar to the proof of Theorem \ref{theorem:1} and is
omitted.
\end{proof}

Note that according to Theorem \ref{theorem:1}, the performance of
SDF algorithm in terms of rate region depends on $\textbf{E}[M(t)]$.
if $\textbf{E}[M(t)]>1$, then $\epsilon
>0$. In other words, we can increase the rate region. Next, we
determine the the value of $\textbf{E}[M(t)]$ when channels are
homogenous and heterogenous.

\section{Performance of SDF with Homogenous Channels}
We begin  with a Lemma that shows  how to determine the value of
$\textbf{E}[M(t)]$ when channels are homogenous.

\begin{lemma}
\label{lemma:1} When channels are homogenous, $\textbf{E}[M(t)]$ is
given as follows:
\begin{align}
&\textbf{E}[M(t)]=\notag\\
&\left[p_1 + p_2(\sum_{k=2}^L p_k) + p_3(\sum_{k=3}^L p_k) +\ldots +
p_L^2\right](N-1)\label{eq:M_gen}
\end{align}
\end{lemma}

\begin{proof}
 Let $\textbf{E}[M(t)| R_{i^*}(t)=r_k, n=i^*]$ be the
conditional expectation when the user which has the maximum queue
length and its CSI are given. Note that for homogenous channels, the
following equality holds,
\begin{align}
\textbf{E}[M(t)| R_{i^*}(t)=r_k, n=i^*]=\textbf{E}[M(t)|
R_{i^*}(t)=r_k],
\end{align}
and $\textbf{E}[M(t)]$ is determined as follows:
\begin{align*}
\textbf{E}[M(t)]=\left[\sum_{k=1}^L \textbf{E}[M(t) |
R_{i^*}(t)=r_k]\Pr[R_{i^*}(t)=r_k]\right].
\end{align*}
When $R_{i^*}(t)=r_1$, then SDF determines the user which has the
maximum weighted rate  by only using one fraction of time slot with
probability $p_1$, i.e., $\Pr[R_{i^*}(t)=r_k]=p_1$. If this event
occurs, the other $N-1$ users  do not report their CSIs with
probability 1 and $\textbf{E}[M(t)|R_{i^*}(t)=r_1]=(N-1)$.
Similarly, when $R_{i^*}(t)=r_2$, user $j\neq i^*$ does not report
its CSI if $R_j(t)\leq r_2$, and $\Pr[R_j(t) \leq r_2]=\sum_{k=2}^L
p_k$. Hence, with homogenous channels,
$\textbf{E}[M(t)|R_{i^*}(t)=r_2]=(\sum_{k=2}^L p_k)(N-1)$. Now, for
a given any value of $L$, we give a general formulation for
$\textbf{E}[M(t)]$ as follows,
\begin{align*}
&\textbf{E}[M(t)]=\notag\\
&\left[p_1 + p_2(\sum_{k=2}^L p_k) + p_3(\sum_{k=3}^L p_k) +\ldots +
p_L^2\right](N-1)
\end{align*}
\end{proof}

Next, we investigate the case when the channels are homogenous and
\textit{uniformly} distributed such that $p_k=\frac{1}{L}$ for all
$k$.
\begin{lemma}
\label{lemma:2} When channels are homogenous and uniformly
distributed, $\textbf{E}[M(t)]$ is given as follows:
\begin{align}
\textbf{E}[M(t)]=\left[\frac{1}{L} + \frac{L(L-1)}{2L^2}
\right](N-1)=\left[\frac{1}{2} + \frac{1}{2L}
\right](N-1)\label{eq:EM2}
\end{align}
\end{lemma}

\begin{proof}
If all channels are uniformly distributed, then $p_k=\frac{1}{L}$.
In that case, by using \eqref{eq:M_gen} $\textbf{E}[M(t)]$ is given
by,
\begin{align*}
\textbf{E}[M(t)]&=\\&\left[\frac{1}{L} + \frac{1}{L}(\sum_{k=2}^L
\frac{1}{L}) + \frac{1}{L}(\sum_{k=3}^L \frac{1}{L}) + \ldots +
\frac{1}{L^2}\right](N-1)
\end{align*}
Thus, we have,
\begin{align}
\textbf{E}[M(t)]&=\\&\left[\frac{1}{L} + \frac{1}{L^2}(L-1) +
\frac{1}{L^2}(L-2) + \ldots +
\frac{1}{L^2}\right](N-1)\label{eq:EM1}
\end{align}
When we arrange \eqref{eq:EM1}, we obtain
\begin{align*}
\textbf{E}[M(t)]=\left[\frac{1}{L} + \frac{L(L-1)}{2L^2}
\right](N-1)=\left[\frac{1}{2} + \frac{1}{2L}
\right](N-1)\label{eq:EM2}
\end{align*}
\end{proof}

\begin{theorem}
\label{thm:main} When channels are homogenous and if $N>3$ we
guarantee to expend the rate region, i.e., $\epsilon>0$ and the
amount of increase in rate region is equal to $\epsilon$ given in
\eqref{eq:epsilon}.
\end{theorem}


\begin{proof}
The proof of the theorem can be done in two parts. In the first
part, we show that $\textbf{E}[M(t)]$ is minimum when channels are
uniform. In the second part of the proof, we show that if $N>3$,
then $\epsilon>0$ when channels are uniform. We begin  by proving
the first part of the proof.

We first show that $\textbf{E}[M(t)]$ is jointly convex function of
$p_1,p_2,\ldots,p_L$ and then, the minimum of this convex function
is achieved when channels are uniform.

\begin{lemma}
\label{lemma:3} $\textbf{E}[M(t)]$ is jointly convex function of
$p_1,p_2,\ldots,p_L$.
\end{lemma}

\begin{proof}
By using $p_1=1-\sum_{k=2}^L p_k$, the Hessian of $\textbf{E}[M(t)]$
in \eqref{eq:M_gen} can be given as follows,
\begin{align*}
H= \left(\begin{array}{c c c c c c c}
 2 & 1 & 1 & 1 & 1 & \cdots & 1 \\
 1 & 2 & 1 & 1 & 1 & \cdots & 1 \\
 1 & 1 & 2 & 1 & 1 & \cdots & 1 \\
 1 & 1 & 1 & 2 & 1 & \cdots & 1 \\
 1 & 1 & 1 & 1 & 2 & \cdots & 1 \\
 \vdots  & \vdots & \vdots  & \vdots & \vdots  & \ddots & \vdots  \\
 1 & 1 & 1 & 1 & 1 & \cdots & 2 \\
  \end{array} \right)(N-1)
\end{align*}
Now, we show that $H$ is positive definite matrix. Let
$\textbf{x}=[x_1 \ x_2 \ x_3 \ldots \ x_L]$ be any vector and
$\textbf{x} \in \mathbb{R}^{L-1}$. If $\textbf{x}^T H \textbf{x} >0
$ then, $H$ is positive definite matrix and $\textbf{E}[M(t)]$ is
convex function of $p_1,p_2,\ldots,p_L$ \cite{Boyd:Convex04}.
\begin{equation}
\textbf{x}^T H \textbf{x}= \left[\sum_{l=1} x_l^2 + \left(\sum_{l=1}
x_l\right)^2\right](N-1)
> 0.
\end{equation}
Hence, $\textbf{E}[M(t)]$ is jointly convex function of
$p_1,p_2,\ldots,p_L$.
\end{proof}

\begin{lemma}
\label{lemma:4} $\textbf{E}[M(t)]$ has the minimum value when
channels are uniformly distributed.
\end{lemma}
\begin{proof}
We already showed that $\textbf{E}[M(t)]$ is jointly convex function
of $p_1,p_2,\ldots,p_L$. Hence, by using $p_1=1-\sum_{k=2}^L p_k$ we
have the following $L-1$ linear equations,
\begin{align}
\frac{\delta \textbf{E}[M(t)]}{\delta p_2}&=-1 + 2p_2 + (p_3 + p_4 + \ldots + p_L)=0\\
\frac{\delta \textbf{E}[M(t)]}{\delta p_3}&=-1 + 2p_3 + (p_2 + p_4 + \ldots + p_L)=0\\
\frac{\delta \textbf{E}[M(t)]}{\delta p_4}&=-1 + 2p_4 + (p_2 + p_3 + p_5 + \ldots + p_{L})=0\\
\vdots\\
\frac{\delta \textbf{E}[M(t)]}{\delta p_L}&=-1 + 2p_{L} + (p_2 + p_3 + p_5 + \ldots + p_{L-1})=0\\
\end{align}
with matrix notation,

\begin{align}
\left(\begin{array}{c c c c c c c}
 2 & 1 & 1 & 1 & 1 & \cdots & 1 \\
 1 & 2 & 1 & 1 & 1 & \cdots & 1 \\
 1 & 1 & 2 & 1 & 1 & \cdots & 1 \\
 1 & 1 & 1 & 2 & 1 & \cdots & 1 \\
 1 & 1 & 1 & 1 & 2 & \cdots & 1 \\
 \vdots  & \vdots & \vdots  & \vdots & \vdots  & \ddots & \vdots  \\
 1 & 1 & 1 & 1 & 1 & \cdots & 2 \\
  \end{array} \right)\left(\begin{array}{c}
 p_2\\
p_3\\
 p_4 \\
  p_5 \\
\vdots  \\
p_L\\
  \end{array} \right)=\left(\begin{array}{c}
1\\
1\\
1\\
1\\
\vdots  \\
1\\
  \end{array} \right)
\end{align}
Solving this linear system, we have,
\begin{equation}
p_k-p_l=0, \forall k,l \ k\neq l
\end{equation}
Hence,
\begin{equation}
p_k=p_l, \forall k,l \ k\neq l
\end{equation}
Thus,
\begin{equation}
p_k=\frac{1}{L}, \forall k.
\end{equation}
Thus, when the channel distributions are uniform, $\textbf{E}[M(t)]$
has the minimum value.
\end{proof}

We now prove the second part of the theorem. We  show that when
channels are uniform, $\textbf{E}[M(t)]$ is a decreasing function of
$L$.

\begin{lemma}
\label{lemma:5} $\textbf{E}[M(t)]$ is a decreasing function of $L$.
\end{lemma}
\begin{proof}
From \eqref{eq:EM2},
\begin{equation}
\textbf{E}[M(t)]=\left[\frac{1}{2} + \frac{1}{2L} \right](N-1)
\end{equation}
Taking the derivative of $\textbf{E}[M(t)]$ with respect to $L$
yields that,
\begin{equation}
\frac {d \textbf{E}[M(t)]}{dL}=\left[\frac{-1}{2L^2} \right](N-1)
<0.
\end{equation}
Thus, $\textbf{E}[M(t)]$ is a decreasing function of $L$.
\end{proof}

Now, it is easy to see that in \eqref{eq:EM2}, taking $L \rightarrow
\infty$ yields that
\begin{equation}
\lim_{L \rightarrow\infty} \textbf{E}[M(t)]=\lim_{L
\rightarrow\infty} \left[\frac{1}{2} +
\frac{1}{2L}\right](N-1)=\frac{N-1}{2}.
\end{equation}
In the limiting case, when $N>3$, $\textbf{E}[M(t)]>1$. In addition,
according to Lemma \ref{lemma:5} if $L$ is finite,
$\textbf{E}[M(t)]$ is still greater than 1 whenever $N>3$ since
$\textbf{E}[M(t)]$ decreases as $L$ increases. We can conclude that
when all channels are uniformly distributed, and when $N>3$, then
$\textbf{E}[M(t)]>1$. As a result, we guarantee to expend the rate
region, i.e., $\epsilon>0$. In addition, according to Lemma
\ref{lemma:3} and Lemma \ref{lemma:4}, $\textbf{E}[M(t)]$ has its
minimum value when channels are uniform. Therefore, for homogenous
channels, when $N>3$, then $\epsilon>0$ and rate region is expended.
This completes the proof.
\end{proof}

\section{Performance of SDF with Non-homogenous channel}
When the channels are not identical, then $p_{nk}\neq p_{mk}$, where
$n \neq m$ and $\forall k$. In addition, for homogenous channels,
the following equality holds,
\begin{align}
\textbf{E}[M(t)| R_{i^*}(t)=r_k, n=i^*]=\textbf{E}[M(t)|
R_{i^*}(t)=r_k],
\end{align}
However, this condition cannot be hold when channels are
heterogenous.  Hence, $\textbf{E}[M(t)]$ cannot be determined in a
similar way  of homogenous channels. We have the following results
for heterogenous channels.

\begin{theorem}
\label{thm:main_het} When channels are non-homogenous and  as $L$
goes infinity, i.e., $L\rightarrow \infty$, or $N$ is sufficiently
large, then $\epsilon > 0$ i.e., we guarantee to expand the rate
region.
\end{theorem}

\begin{proof}
 the proof is provided in Appendix \ref{sec:main_het}.
\end{proof}


\section{Numerical Results}
\label{sec:sim} In our simulations, we model a single cell CDMA
downlink transmission utilizing high data rate (HDR) \cite{HDR1}.
The base station serves 20 users and keeps a separate queue for each
user. Time is slotted with length $T_s=5$ ms. Packets arrive at each
slot according to Poisson distribution for each users with mean
$\lambda_n$. The size of a packet is set to 128 bytes which
corresponds to the size of an HDR packet. Each channel has 5
possible states with rates as given in Table \ref{tab:table1}.
\begin{table}[h!]
\renewcommand{\arraystretch}{1.5}
\begin{center}
    \begin{tabular}{ | l | l | l | l | l | l |}
    \hline
    Rates & $r_1$ & $r_2$ & $r_3$ & $r_4$ & $r_5$ \\ \hline
    kb/s & 1843.2 & 1228.8 & 614.4 & 307.2 & 76.8   \\
    \hline
    \end{tabular}
    \caption{Possible Physical Rates}
    \label{tab:table1}
\end{center}
\end{table}

\subsection{Homogenous Channels}
First, we evaluate the performance of SDF algorithm when channels
are homogenous and both uniform and non-uniform distributions.
\subsubsection{Uniform Channels}
In uniform case, $p_n=0.2$ for all users since there are 5 channel
states and the channel state distributions are identical.
Figure~\ref{fig:fig1} depicts the average total queue sizes in terms
of packets vs. the overall arrival rate when $\beta=0.01$ and
$\beta=0.02$. The maximum supportable arrival rate is achieved in
\textit{hypothetical} case where the probing cost is zero, i.e.,
$\beta=0$ and the minimum supportable rate is achieved when full CSI
is obtained and $\beta=0.02$. When $\beta=0.01$ and SDF algorithm is
applied, we are closest to the hypothetical rate region. Clearly,
SDF algorithm outperforms the full CSI Max-Weight algorithm in terms
of rate  region for both $\beta=0.01$ and $\beta=0.02$.
\begin{figure}[t]
     \centering
     \includegraphics[width=1.07\columnwidth]{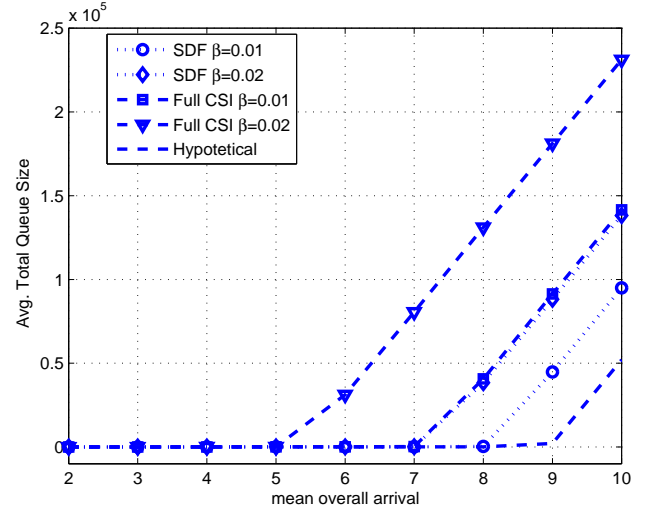}
     \caption{Performance of SDf algorithm with Homogenous and Uniform channels.}
     \label{fig:fig1}
\end{figure}

\subsubsection{Non-Uniform Channels}
Here, we investigate the performance of SDf algorithm when channels
are identical but the channel state distributions are not uniform.
In this case, the channel state distribution is given in Table
\ref{tab:table2}:
\begin{table}[h!]
\renewcommand{\arraystretch}{1.5}
\begin{center}
    \begin{tabular}{ | l | l | l | l | l | l |}
    \hline
     & $p_1$ & $p_2$ & $p_3$ & $p_4$ & $p_5$ \\ \hline
    probabilities & 0.3 & 0.3 & 0.2 & 0.1 & 0.1   \\
    \hline
    \end{tabular}
    \caption{Channel state distribution}
    \label{tab:table2}
\end{center}
\end{table}

Figure~\ref{fig:fig2} depicts the average total queue sizes in terms
of packets vs. the overall arrival rate when channel state
distributions are uniform and as given in Table \ref{tab:table2}
which is defined as non-uniform and  $\beta=0.02$. The maximum
supportable arrival rate is achieved in hypothetical case the
minimum supportable rate is achieved when full CSI is obtained. It
is easy to see that the maximum supportable rate achieved by SDF
algorithm when both uniform and non-uniform channels are the same.
\begin{figure}[t]
     \centering
     \includegraphics[width=1.07\columnwidth]{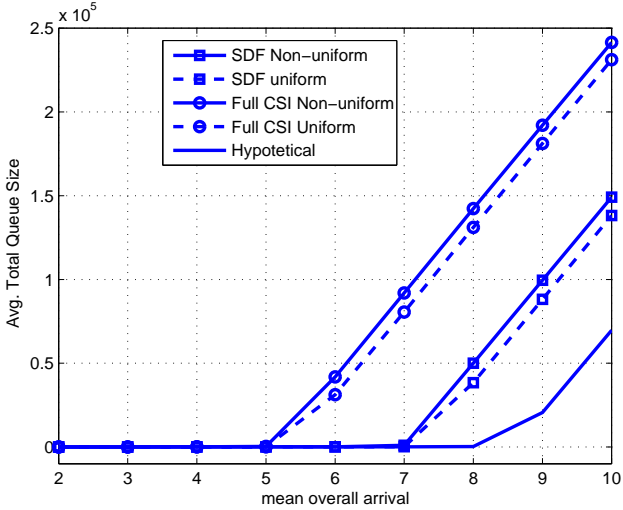}
     \caption{Performance of SDf algorithm with Homogenous and Non-Uniform channels.}
     \label{fig:fig2}
\end{figure}

\subsection{Heterogenous Channels}
Here, we investigate the performance of SDF algorithm when channels
are neither homogenous nor uniform. For the heterogeneous case, we
divide the users into four groups where there are 5 users in each
gruop. The channel state distributions are given for each group in
Table \ref{tab:table3}.
\begin{table}[t]
\renewcommand{\arraystretch}{1.5}
\begin{center}
    \begin{tabular}{ | l | l | l | l | l | l |}
    \hline
     & $p_1$ & $p_2$ & $p_3$ & $p_4$ & $p_5$ \\ \hline
    Group 1 & 0.3 & 0.3 & 0.2 & 0.1 & 0.1   \\ \hline
    Group 2 & 0.1 & 0.1 & 0.2 & 0.3 & 0.3  \\ \hline
    Group 3 & 0.25 & 0.15 & 0.1 & 0.25 & 0.15  \\ \hline
    Group 4 & 0.15 & 0.25 & 0.25 & 0.1 & 0.25  \\
    \hline
    \end{tabular}
    \caption{Channel state distribution with heterogenous channels}
    \label{tab:table3}
\end{center}
\end{table}

\begin{figure}[t]
     \centering
     \includegraphics[width=1.07\columnwidth]{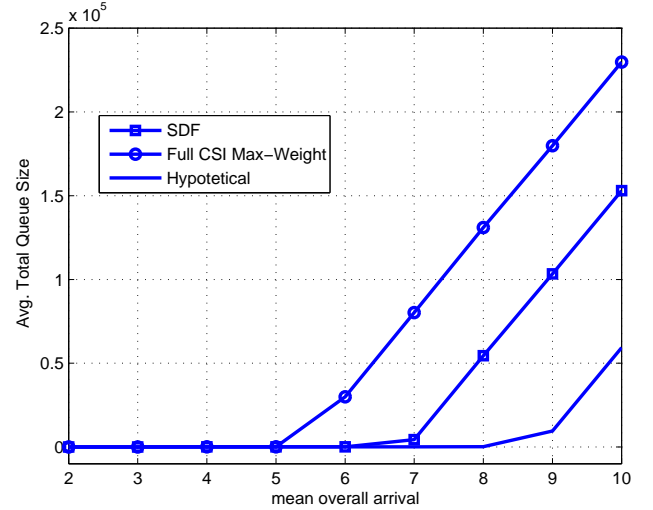}
     \caption{Performance of SDf algorithm with Heterogenous channels.}
     \label{fig:fig3}
\end{figure}
Figure~\ref{fig:fig3} depicts the average total queue sizes in terms
of packets vs. the overall arrival rate when channel state
distributions are given as in Table \ref{tab:table3} and
$\beta=0.02$. The maximum supportable arrival rate still is achieved
in hypothetical case whereas the minimum supportable rate is
achieved when full CSI is obtained. It is easy to see that the
maximum supportable rate achieved by SDF algorithm when uniform,
non-uniform and heterogenous channels are the same.

\subsubsection{Non-Stationary Channels}
Next, we demonstrate the performance of SDF algorithm over
time-correlated and non-stationary channels. The channel between the
BS and each user is modeled as a \textit{correlated} Rayleigh fading
according to Jakes' model with different Doppler frequencies varying
randomly between 5 Hz and 15 Hz. We set $\textrm{BW}=1.25$ MHz and
$\textrm{SNR}=10$ dB.  Let $H_n(t)$ denote CSI of user $n$ at time
slot $t$. $H_n(t)$ is a random process which does not have a
stationary probability distribution, i.e, the mean of the channel
gain changes over time. Let $h_n(t)$ represent the realization of
$H_n(t)$ at time $t$, $n\in \{1,2,\dots,N\}$. Then, the maximum
number of bits of a user may transmit is given as,
\begin{align}
R_n(t)=T_s(1-\beta N_p(t))\textrm{BW}\log_2\left(1 +
\textrm{SNR}\times h_n(t)\right)
\end{align}
where $BW$ is the bandwidth of a channel. Similar to the previous
scenario, Figure~\ref{fig:nonsta} depicts the average total queue
sizes in terms of packets vs. the overall arrival rate when the
channels are non-stationary and $\beta=0.02$.  The maximum and the
minimum arrival rates while keeping the queues stabile are achieved
in hypothetical and full CSI cases, respectively. As the overall
arrival rate exceeds 14 packets/slot queue sizes suddenly increase
with full CSI Max-Weight and the network becomes unstable. However,
SDF  improves over full CSI Max-Weight by supporting the overall
arrival rate of up to 16 packets/slot. In other words, SDF algorithm
is able to sustain 14\% more traffic than the full CSI Max-Weight
algorithm. Thus, SDF appears to stabilize a larger range of arrival
rates.
\begin{figure}[t]
     \centering
     \includegraphics[width=1.07\columnwidth]{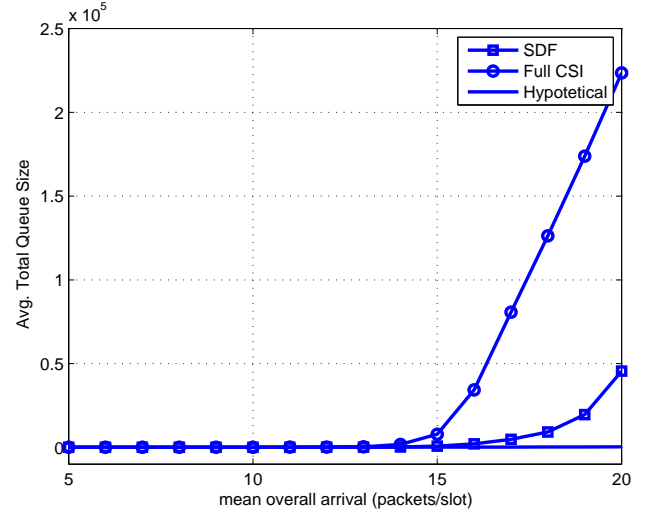}
     \caption{Performance of SDf algorithm over non-stationary channels.}
     \label{fig:nonsta}
\end{figure}
\section{Conclusion}
\label{sec:conclusion} In this paper, we have considered the joint
scheduling and channel probing problem in a single channel wireless
downlink network. We have developed a low complex and dynamic
feedback algorithm named SDF. We have shown that SDF algorithm can
support $1+\epsilon$ fraction of full achievable rate region. Then,
we have proved the sufficient condition for $\epsilon>0$. Our
simulation results demonstrated a significant performance gain of
the proposed algorithm compared to the case when the full CSI is
obtained. In this extended abstract, the analytical and simulation
results are given by assuming a single channel wireless network. In
the full version of the paper, we will provide the result by
considering the multichannel wireless system, i.e., OFDM networks.

\bibliographystyle{IEEEtran}
\bibliography{IEEEabrv,ref}

\appendices
\section{Proof of Theorem \ref{thm:main_het}}
\label{sec:main_het}
\subsection{SDF with Non-homogenous channel}
When the channels are not identical, i.e., $p_{nk}\neq p_{mk}$,
where $n \neq m$ and $\forall k$. We define the $\chi_n$  such that
$n=i^*$, i.e., user $n$ is the user with maximum queue size at a
time. In addition, we define $\varphi_k$ such that $R_{i^*}(t)=r_k$,
i.e., the user with maximum queue sizer is at channel state $k$ at
time $t$. Hence,
\begin{align*}
&\Pr[\chi_n]=\Pr[n=i^*],\\
&\Pr[\varphi_k]=\Pr[R_{i^*}(t)=r_k].
\end{align*}
Then $\textbf{E}[M(t)]$ can be found as follows:
\begin{equation}
\textbf{E} [M(t)]=\sum_{n=1}^N \sum_{k=1}^L \textbf{E}[M(t)
|\chi_n,\varphi_k ]Pr[\chi_n, \varphi_k].
\end{equation}
Hence,
\begin{align*}
&\textbf{E}[M(t)]=\\&\Pr[\chi_1]\left(p_{11}(N-1) +
p_{12}(\sum_{n=2}^N \sum_{k=2}^L p_{nk}) + \ldots +
p_{1L}(\sum_{n=2}^N  p_{nL})\right)\notag\\+
&\Pr[\chi_2]\left(p_{21}(N-1) + p_{22}(\sum_{\substack{n=1\\ n \neq
2}}^N \sum_{k=2}^L p_{nk}) +\ldots + p_{2L}(\sum_{\substack{n=1\\ n
\neq 2}}^N p_{nL})\right)\notag\\+ &Pr[\chi_3]\left(p_{31}(N-1) +
p_{32}(\sum_{\substack{n=1\\ n \neq 3}}^N \sum_{k=2}^L p_{nk})
+\ldots + p_{3L}(\sum_{\substack{n=1\\ n \neq 3}}^N
p_{nL})\right)\notag\\
 \vdots \\ +
&\Pr[\chi_N]\left(p_{N1}(N-1) + p_{N2}(\sum_{\substack{n=1\\
n \neq N}}^N \sum_{k=2}^L p_{nk}) + \ldots + p_{NL}(\sum_{\substack{n=1\\
n \neq N}}^N p_{nL})\right).
\end{align*}
Note that $\Pr[\chi_n] \geq p^{min}_q$ for all $n$, where $0 <
p^{min}_q <1$. Hence, a lower bound on $\textbf{E}[M(t)]$ can be
given as follow,
\begin{align*}
\textbf{E}[M(t)] \geq  Np^{min}_q [&p^{min}(N-1)+ p^{min}(p^{min}
(L-1)(N-1)\\& + p^{min}(p^{min} (L-2)(N-1) + \ldots \\&+
p^{min}(p^{min} (N-1) ]
\end{align*}
By rearranging, we have,
\begin{align*}
\textbf{E}[M(t)] \geq Np^{min}_q \left[ (N-1) \left(p^{min} +
(p^{min})^2 \left[\frac{L(L-1)}{2}\right]\right) \right]
\end{align*}
Therefore,  we can guarantee to achieve a larger capacity region
when the following condition is satisfied:
\begin{align*}
N(N-1) > \frac{1}{p^{min}_q \left[\left(p^{min} + (p^{min})^2
\left[\frac{L(L-1)}{2}\right]\right) \right]}
\end{align*}
Clearly, $L\rightarrow \infty$ or $N$ is large enough,  this
condition holds and $\epsilon > 0$. This completes the proof.
\end{document}